\documentclass[12pt,a4paper]{article}
\usepackage{graphicx}
\usepackage{amsmath, amsfonts, amssymb, amsthm} 
\usepackage[symbol]{footmisc}

\newtheorem*{theorem*}{Theorem}
\newtheorem*{corollary*}{Corollary}

\usepackage[top=1.5in, bottom=1.5in, left=1.23in,
right=1.23in]{geometry}

\begin{document}

\thispagestyle{empty}  

\centerline{On the Free Will Theorem}
\bigskip
\centerline{Simon Kochen}
\bigskip
\medskip

We discuss a form of the Free Will Theorem which is simpler and more
general than the original in Conway and Kochen~\cite{C-K}.  One reason that it is
more transparent is that the new proof is invariant, and does not use
Lorentz frames, as  [1] did.  The main difference is that the original
theorem was conditional:  the particle's spontaneity was based on the
free will of the experimenter.  The new theorem is absolute, since it
does not require any use of free will on the experimenter's part.  The
theorem also strengthens the Kochen-Specker Theorem [3], which shows
that non-contextual deterministic theories cannot predict the results
of certain experiments.  By using EPR correlations between these
experiments, we show that even contextual deterministic theories are
unable to predict the experimental result.

We first recall the two axioms from [1] which we shall use.  For a spin
$1$ system, the projections $S_x^2$, $S_y^2$, $S_z^2$ commute, and
$S_x^2 + S_y^2 + S_z^2 = 2$, so the triple $(S_x^2, S_y^2, S_z^2)$ has
the eigenvalues $(1,0,1)$, $(0,1,1)$, or $(1,1,0)$.

The projection $S_z^2$ can be measured by using an inhomogeneous
electric field in place of the magnetic field in the Stern-Gerlach
experiment.  Note that an electric field has a polar field vector, which
does not distinguish between positive and negative spin values, and so
measures $S_z^2$, but not $S_z$.  Such an experiment was carried out by
E.~Wrede in~\cite{W}.  

We may measure $S_x^2$, $S_y^2$, and $S_z^2$ successively,  using
a  Wrede apparatus, without disturbance, since they commute. We 
may also measure them simultaneously by measuring $S_x^2 - S_y^2$, which has the
eigenvalues $1$, $-1$ and $0$, implying that $S_x^2$, $S_y^2$, and
$S_z^2$ has eigenvalue $0$ respectively.

This gives us the 
\medskip

\noindent {\bf SPIN Axiom:} Measurements of the squared components of spin $(S_x^2, S_y^2, S_z^2)$
of a spin $1$ particle in three orthogonal directions $x$, $y$, $z$ give
the results $(1,0,1)$, $(0,1,1)$, or $(1,1,0)$.  
\medskip

For the second axiom we use a form of EPR for two spin $1$ particles
with total spin $0$.  Then $S_z$ has eigenvalues $1$, $-1$, $0$ for the
first particle if and only if $S_z$ has eigenvalues $-1$, $1$, $0$ for
the second, so $S_z^2$ has the same eigenvalue for the two particles.
The Twin axiom is a version of this result.
\medskip

\noindent {\bf TWIN Axiom:}
For two spin $1$ particles $a$ and $b$ of total spin zero, suppose
experimenter $A$ performs the triple experiment of measuring $(S_x^2,
S_y^2, S_z^2)$ on $a$, while experimenter $B$ performs the single 
experiment of measuring $S_w^2$ on $b$.  Then, if $w$ is in the same
direction as one of $x$, $y$, and $z$, experimenter $B$'s measurement
will yield the same answer as the corresponding measurement of $A$.
\medskip

We call the measurement of the three projections $(S_x^2, S_y^2,
S_z^2)$ a {\em triple experiment\/} and the measurement of the single
projection $S_w^2$, a {\em single experiment\/}.

We assume a locality result from special relativity.
\bigskip

\noindent {\bf FIN Axiom:}
The speed with which the signal that the measured value $(1,0,1)$,
$(1,1,0)$, or $(0,1,1)$ of the triple $(S_x^2, S_y^2, S_z^2)$ can be
transmitted has an upper bound $c$, the speed of light in a 
vacuum.
\medskip

We shall use a mathematical result from Kochen and Specker~\cite{K-S}.  A
version of the theorem asserts that we can specify 33 directions for
which there is no function into $(0,1)$ with the value $(1,0,1)$,
$(1,1,0)$, or $(0,1,1)$ on every one of the $40$ orthogonal triples
formed out of the 33 directions.

The 33 directions are the axes of symmetry of the three cubes of the
$3$-cube polyhedron, illustrated below, obtained by rotating a
coordinate cube by $45{}^\circ$ about the three coordinate
axes.\footnote[1]{The original theorem used in 117 directions;
the smallest number known is 31 (by Conway and Kochen).  I am using the
33 directions because of their simple geometric description. For a proof
see Peres~\cite{P}; for a co-ordinate-free proof see Conway and Kochen [1].
It has been proved that at least 21 directions are needed for the
theorem.}\par

 \begin{figure*}[h]
\centering \includegraphics{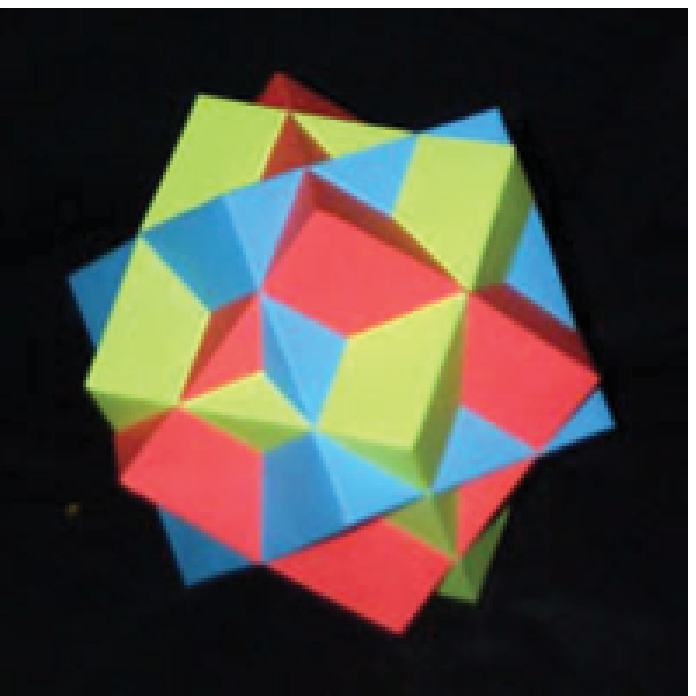}
 \end{figure*}
 \bigskip

We describe a specific ordering, the {\em lexical ordering\/}, of the
$40$
orthogonal triples of the 33 axes of symmetry.  We first order the 13
axes of symmetry of the co-ordinate cube as follows:
\[
100,010,001,110,101,10\underline{1},1\underline{1}0,011,01\underline{1},111,11\underline{1},1\underline{1}1,\underline{1}11.
\]
Here $\underline{1}$ means $-1$ and $abc$ denotes the axis joining the
point $(a,b,c)$ to the origin. We then use this order on each of the three 
rotated cubes of the $3$-cube polyhedron, and order the three cubes by 
giving order priority to the rotation of $45{}^\circ$ about the $x$ axis 
over the rotation about the $y$ axis, and the $y$ axis rotation over 
the $z$ axis rotation.

We now eliminate the six axes that are repetitions of their earlier
appearances in the ordering.  This gives an ordering of the 33 axes of
symmetry of the three cubes of the $3$-cube.  Finally, we order the $40$
orthogonal triples of the 33 axes by giving priority to the first
co-ordinate over the second, and the second over the third.  The result
is the lexical ordering of the $40$ orthogonal triples of symmetry axes of
the cube:
$(\mathbf{x}_1,\mathbf{y}_1,\mathbf{z}_1),\ldots,(\mathbf{x}_{40},\mathbf{y}_{40},\mathbf{z}_{40})$.  

We shall use the following remarks in the proof.  

A point in space-time is given by a $4$-vector $\rho =
(\mathbf{r}\negthinspace ,ct)$.
Every point $\rho_1$ is the vertex of a light cone, defined by all such
$\rho$ such that $|\rho-\rho_1|^2 = 0$, i.e.
\[
|\mathbf{r}-\mathbf{r}_1|^2 - (ct-ct_1)^2 = 0.
\]
\noindent The light cone with vertex at a given space-time point is
invariant, and does not depend on the choice of inertial frame, so
different observers see the same light cone at each point.

We shall say that a point lies {\em in\/} a light cone to mean that the
point lies on the light cone or in its interior.

The occurrence of an event has extension in space-time.  As is usual, we
shall idealize the occurrence of an event to a point in space-time.  Our
proof remains valid without this idealization.

We call this result below the {\em Unconditional\/} (or {\em Absolute\/)}
Free Will Theorem

\begin{theorem*}
The Spin, Twin, and Fin axioms imply that the particle's response to the
$40$ triple experiments is {\em free\/}, in the sense that it is not
determined by the events in the past light cone of that event.
\end{theorem*}

\begin{proof}
Two spin $1$ particles $a$ and $b$ have been prepared to have total spin
$0$.\phantom\qedhere
\end{proof}

Experimenter A performs one of the $40$ triple experiments on particle $a$
and  B  performs one of the 33 single experiments on particle $b$.

Assume, contrary to the theorem, that particles $a$ and $b$'s responses
are functions $\theta_a(\alpha)$ and $\theta_b(\beta)$ of properties
(events) $\alpha$ and $\beta$ ranging in their respective past light
cones.  Thus, 

\begin{align*}
\theta_a(\alpha) &= (1,0,1) (0,1,1),\   \mathrm{or}\   (1,1,0)\\
\theta_b(\beta) &= 1\   \mathrm{or}\    0
\end{align*}

1.  We shall first show how to reduce the domains of
$\theta_a$ and $\theta_b$.

One of the properties on which $b$'s response depends is the
disposition of B's electromagnetic to give a field which is
inhomogeneous in the direction $\mathbf{w}$.  (In order to measure
$S_w^2$ for particle $b$).  We shall denote this property in short by
the vector $\mathbf{w}$.   Similarly, we denote the properties of A's
electromagnet in directions $\mathbf{x}$, $\mathbf{y}$, and $\mathbf{z}$
by the vectors $\mathbf{x}$, $\mathbf{y}$, and $\mathbf{z}$.

Aside from this explicit functional dependence on $\mathbf{x}$,
$\mathbf{y}$, and $\mathbf{z}$, and $\mathbf{w}$, $\theta_a$ may also
implicitly be a function of $\mathbf{x}$, $\mathbf{y}$, $\mathbf{z}$,
and $\mathbf{w}$ via properties, $\mu$ $(\mathbf{x}, \mathbf{y},
\mathbf{z}, \mathbf{w})$ which are themselves functions of these
variables.  We can make this implicit dependence of $\theta_a$ explicit
by inserting all such functions in $\theta_a$: 
\[
\theta_a(\mathbf{x},\mathbf{y},\mathbf{z},\mathbf{w},\mu(\mathbf{x},\mathbf{y},\mathbf{z},\mathbf{w});
\alpha'),
\]
where all the $\alpha'$ that remain of the $\alpha$ are not functions of
$\mathbf{x}$, $\mathbf{y}$, $\mathbf{z}$, and $\mathbf{w}$.

To simplify notation, we continue to call this new function
$\theta_a(\mathbf{x},\mathbf{y},\mathbf{z},\mathbf{w};\alpha')$.
Similarly, we obtain a function
$\theta_b(\mathbf{x},\mathbf{y},\mathbf{z},\mathbf{w};\beta')$ where the
$\beta'$ are not functions of $\mathbf{x}$, $\mathbf{y}$, $\mathbf{z}$,
and $\mathbf{w}$.

Now, the $\alpha'$ lie in the past light cone of $a$, and so these
properties will take on specific values $\alpha'_0$ before $a$'s
response, resulting in the function
$\theta_a(\mathbf{x},\mathbf{y},\mathbf{z},\mathbf{w};\alpha'_0)$.  The
$\alpha'_0$'s consist of the constant values $1$ or $0$.  We drop these
constants, and call the resulting function
$\theta_b(\mathbf{x},\mathbf{y},\mathbf{z},\mathbf{w})$. Similarly, we
call $b$'s response function $\theta_b(\mathbf{x}, \mathbf{y}, \mathbf{z},
\mathbf{w})$.

Let $a$ and $b$ lie at the $4$-vectors $\rho_1$ and $\rho_2$
respectively.  We now assume that the point $\rho_2$ lies outside the
past light cone of the point $\rho_1$.  Then $\mathbf{w}$ is not in the
domain of $\theta_a$, and we may drop the $\mathbf{w}$, and write $a$'s
response as $\theta_a(\mathbf{x},\mathbf{y},\mathbf{z})$.

A point $\rho_2$ lies outside the light cone of $\rho_1$ if and only if
$|\rho_1 - \rho_2|^2\ge 0$.  Since $|\rho_1-\rho_2|^2 =
|\rho_2-\rho_1|^2$, $\rho_2$ lies outside the light cone of $\rho_1$ if
and only if $\rho_1$ lies outside the light cone of $\rho_2$.  Thus, $\rho_1$
lies outside the light cone of $\rho_2$.  It follows that $\mathbf{x}$,
$\mathbf{y}$, and $\mathbf{z}$ are not in the domain of $\theta_b$, so
we may write $b$'s response as $\theta_b(\mathbf{w})$.
\smallskip

2.  We now give a description of an apparatus that makes a
measurement that contradicts the prediction given by the
function $\theta_b(\mathbf{w})$.

A computer is connected to a {\em Wrede-type\/}  electromagnet, which
makes three  measurements of $(S_x^2, S_y^2, S_z^2)$ for particle $b$. A
keypad on the computer has $40$ keys, labelled $1,\ldots,40$.  A switch
activates the connection to the electromagnet, so that pressing the
$i^{\mathrm{th}}$ key leads to the rotation of the electromagnet in the
three successive directions $\mathbf{x}_i$, $\mathbf{y}_i$, and
$\mathbf{z}_i$ (where $(\mathbf{x}_i,\mathbf{y}_i,\mathbf{z}_i)$ is the
$i^{\mathrm{th}}$ in the lexical order of the $40$ orthogonal triples), in
order to make three measurements of $(S_{x_i}^2,S_{y_i}^2,S_{z_i}^2)$
of particle $b$.  At the same time, the triple of values
$(\theta_b(\mathbf{x}_i),\theta_b(\mathbf{y}_i),\theta_b(\mathbf{z}_i))$
appears on the screen of the computer.

Now, by the Twin Axiom, for any orthogonal triple
($\mathbf{x},\mathbf{y},\mathbf{z})$,
\begin{equation}
\theta_a(\mathbf{x},\mathbf{y},\mathbf{z}) =
(\theta_b(\mathbf{x}),\theta_b(\mathbf{y}),\theta_b(\mathbf{z})) 
\end{equation}
\noindent The Spin axiom implies that $\theta_a$ has the values
$(1,0,1)$ in some order.  It follows by (1) that $\theta_b$ is a
function on the 33 axes of symmetry that has the values of $(1,0,1)$ in
some order on each of the $40$ orthogonal triples, contradicting the cited
theorem on the $3$-cube.  Thus, there exists an orthogonal triple on
which $\theta_b$ is not $(1,0,1)$, $(1,1,0)$, or $(1,1,0)$.

We shall show that experimenter $B$ will make this contradiction to the
existence of the function $\theta_b$ explicit, without any use of choice
or free will.

To carry out the experiment, $B$ presses the successive keys of the
keypad, starting from $1$.  He stops at the first key, say $k$, for
which the screen shows a triple that is {\em not\/} $(1,0,1)$,
$(0,1,1)$, or $(1,1,0)$, and activates the electromagnet to measure the
triples $(S_{x_k}^2, S_{y_k}^2, S_{z_k}^2)$.

By the Spin Axiom, this measurement must give the triple value
$(1,0,1)$, $(0,1,1)$, or $(1,1,0)$, contradicting the value
given by the function $\theta_b$.

This contradiction shows that no deterministic function
$\theta_b(\mathbf{w})$ or deterministic theory can correctly predict
the results of the $40$ experiments.  In fact, not only a determinist,
but any prognosticator, whether it is a Greek oracle or an omniscient
god, must contradict the results of the 40 experiments.  All that is
required for the contradiction is the existence of the pair of
mathematical functions $\theta_a$ and $\theta_b$ on the axes of
symmetry.

We have been painstakingly explicit in describing this experiment in 
order to make clear that the experimenter $B$ does not make any choices 
in doing the experiment.  In fact, he could be completely replaced by 
a computer.

This shows that no free will on the part of the experimenter is assumed
in this proof.  The original Free Will Theorem assumed the free will
choice by experimenter $B$ of one experiment among the $40$, and so could
not refute a doubter of human free will.

We may extend the theorem to events outside the past light cone to all
points of space-time by invoking causality.

The issue of causality is a hoary philosophical subject, which we shall
not touch upon.  We shall only invoke one generally agreed upon property
of causality, sometimes called the principle of causality, that past
events cannot be the effects of future causes.

We now state the principle in its relativistic version.
\medskip

\noindent {\bf  Relativistic Causality Principle}. 
A cause must precede its effect in all inertial frames.  
\medskip

In the setting of the theorem, since events outside the particle's past
light cone are in the particle's future in some inertial frame, the
Relativistic Causality Principle implies that events outside the
particle's past light cone cannot be the cause of the particle's
response.  The theorem shows that events in the past light cone do not
determine, and a fortiori, do not cause the particle's response.  This
yields the following consequence.

\begin{corollary*}
Spin, Twin, Fin, and the Relativistic Causality Principle imply that the
spin particle's response to the triple experiments is not caused by any
other events.
\end{corollary*}

We end with some consequences of the theorem.
\smallskip

1. A Strengthened Kochen-Specker Theorem.  The KS Theorem shows that
non-contextual deterministic theories contradict the Spin Axiom for a
spin $1$ particle.  By adding EPR, in the form of the Twin Axiom and
using special relativity, the above proof shows that for a spin $1$
particle which is part of two spin $1$ particles of total spin $0$,
assuming contextuality does not help deterministic theories.\par

Contextuality allows the value of $S_z^2$ to vary in he two
triple experiments measuring $(S_x^2, S_y^2, S_z^2)$ and $(S_{x'}^2,
S_{y'}^2, S_z^2)$.  However the Twin Axiom implies that the two values of 
$S_z^2$ are in fact equal, since they are the same as the value of
$S_z^2$ for particle $b$.  Thus, deterministic theories, whether
contextual or non-contextual, cannot predict the measurement results of
the $40$ experiments in this setting of correlated particles.
\smallskip

2. In the proof we have used the correlations of EPR for two spin
particles of total spin $0$ in a relativistic context.  There are
physicists who maintain that EPR yields instantaneous effects rather
than correlations.  We believe this is a mistaken view of EPR, and have
discussed this point in Kochen~\cite{K}.
\smallskip

3. How is our theorem relevant to human free will?  For two thousand
years, philosophers have proposed variations of an argument against the
existence of free will.  In a nutshell, this argument says that both
determinism and randomness are incompatible with free will.  Now the
above theorem shows that determinism is contradicted by the results of
experiments, and so cannot be used against the existence of free will.
On the other hand, randomness has a place in quantum theory.  For
instance, the values of a component of spin of a spin $1/2$ particle,
such as an electron, form a random sequence.  However, as many have
argued, randomness cannot help in establishing free will.  This is
basically because tossing a coin to determine a course of action does
not help in establishing free will. 

We claim that it is not randomness but {\em quantum correlations\/} that
play a crucial role.  For instance, two electrons with random spins
interact in a helium atom, and may form a spin $1$ atom of orthohelium
or a spin $0$ atom of parahelium, neither of which yield random
sequences.  Again, two spin $1$ atoms with total spin $0$ form a system
that is strongly correlated, and is far from yielding random sequences.

Note that in the above proof of our theorem, no pr obability, let alone
randomness, was used at all, only the above correlations (which imply
the Spin and Twin axioms).  It is not simply the randomness of particles,
but correlated probabilities between them that requires the spontaneity
of particles, and hence the possibility of human free will.

More speculatively, if the brain is a quantum neural net (with 
qbits in place of bits), then correlations of the kind used in our
proof would yield the spontaneity that may be the precursor of our
free will.

\end{document}